\documentclass{llncs}
\usepackage[dvipsnames]{xcolor}
\usepackage{makeidx}
\usepackage{amsmath}
\usepackage{amssymb}
\usepackage{gastex}
\usepackage{theorem}
\DeclareMathOperator{\Card}{Card}

\newcommand\CR{\mathcal{CR}}
\newcommand\RR{\mathcal{R}}

\title{Enumeration formul{\ae} in neutral sets}
\author{Francesco Dolce\inst{1} \and Dominique Perrin\inst{1}}
\institute{Universit{\'e} Paris Est, LIGM}

    \def\soft#1{\leavevmode\setbox0=\hbox{h}\dimen7=\ht0\advance
    \dimen7 by-1ex\relax\if t#1\relax\rlap{\raise.6\dimen7
    \hbox{\kern.3ex\char'47}}#1\relax\else\if T#1\relax
    \rlap{\raise.5\dimen7\hbox{\kern1.3ex\char'47}}#1\relax
    \else\if d#1\relax\rlap{\raise.5\dimen7\hbox{\kern.9ex
    \char'47}}#1\relax\else\if D#1\relax\rlap{\raise.5\dimen7
    \hbox{\kern1.4ex\char'47}}#1\relax\else\if l#1\relax
    \rlap{\raise.5\dimen7\hbox{\kern.4ex\char'47}}#1\relax
    \else\if L#1\relax\rlap{\raise.5\dimen7\hbox{\kern.7ex
    \char'47}}#1\relax\else\message{accent \string\soft
    \space #1 not defined!}#1\relax\fi\fi\fi\fi\fi\fi} 

\date{\dateandtime}
\begin{document}

\maketitle

\begin{abstract}
We present several enumeration results holding in sets of words called neutral and which satisfy restrictive conditions on the set of possible extensions of nonempty words.
These formulae concern return words and bifix codes.
They generalize formulae previously known for Sturmian sets or more generally for tree sets.
We also give a geometric example of this class of sets, namely the natural coding of some interval exchange transformations.

\keywords{Neutral Sets, Bifix Codes, Interval Exchanges.}
\end{abstract}

\section{Introduction}
\label{sec:intro}
Sets of words of linear complexity play an important role in combinatorics on words and symbolic dynamics.
This family of sets includes Sturmian sets, interval exchange sets and primitive morphic sets, that is, sets of factors of fixed points of primitive morphisms.

We study here a family of sets of linear complexity, called neutral sets.
They are defined by a property of a graph $E(x)$ associated to each word $x$, called its extension graph and which expresses the possible extensions of $x$ on both sides by a letter of the alphabet $A$.
A set $S$ is neutral if the Euler characteristic of the graph of any nonempty word is equal to $1$, as for a tree.
The Euler characteristic of the graph $E(\varepsilon)$ is called the characteristic of $S$ and is denoted $\chi(S)$.
These sets were first considered in~\cite{BalkovaPelantovaSteiner2008} and in~\cite{acyclicconnectedandtreesets}.
The factor complexity of a neutral set $S$ on $k$ letters is for $n \neq 1$
\begin{equation}
p_n = n(k-\chi(S)) + \chi(S).
\label{eq:complexity}
\end{equation}

We prove here several results concerning neutral sets.
The first one (Theorem~\ref{theo:cardinality}) is a formula giving the cardinality of a finite $S$-maximal bifix code of $S$-degree $n$ in a recurrent neutral set $S$ on $k$ letters as
\begin{equation}
\Card(X) = n(k-\chi(S)) + \chi(S).
\label{eq:cardbifix}
\end{equation}

The remarkable feature is that, for fixed $S$, the cardinality of $X$ depends only on its $S$-degree.
In the particular case where $X$ is the set of all words of $S$ of length $n$, we recover Equation~\eqref{eq:complexity}.
Formula~\eqref{eq:cardbifix} generalizes the formula proved in~\cite{BerstelDeFelicePerrinReutenauerRindone2012} for Sturmian sets and in~\cite{BertheDeFeliceDolceLeroyPerrinReutenauerRindone2013b} for neutral sets of characteristic $1$.

The second one concerns return words.
The set of right first return words to a word $x$ in a factorial set $S$, denoted $\RR_S(x)$, is an important notion.
It is the set of words $u$ such that $xu$ is in $S$ and ends with $x$ for the first time.
In several families of sets of linear complexity, the set of first return words to $x$ is known to be of fixed cardinality independent of $x$.
This was proved for Sturmian words in~\cite{JustinVuillon2000}, for interval exchange sets in ~\cite{Vuillon2007} and for neutral sets of characteristic zero in~\cite{BalkovaPelantovaSteiner2008}.

We first prove here  (Theorem~\ref{theo:completereturn}) that the set $\CR_S(X)$ of complete first return words to a bifix code $X$ in a uniformly recurrent neutral set $S$ on $k$ letters satisfies $\Card(\CR_S(X)) = \Card(X) + k - \chi(S)$.
The remarkable feature here is that, for fixed $S$, the cardinality of $\CR_S(X)$ depends only on $\Card(X)$.
When $X$ is reduced to one element $x$, we have $\CR_S(x) = x\RR_S(x)$ and we recover the result of~\cite{BalkovaPelantovaSteiner2008}.
When $X = S \cap A^n$, then $\CR_S(X) = S \cap A^{n+1}$.
This implies $p_{n+1} = p_n + k - \chi(S)$ and also gives Equation~\eqref{eq:complexity} by induction on $n$.
The proofs of these formul\ae \ use a probability distribution naturally defined on a neutral set.

A third result concerns the decoding of a neutral set by a bifix code.
We prove that the decoding of any recurrent neutral set $S$ by an $S$-maximal bifix code is a neutral set.
This property is proved for uniformly recurrent tree sets in~\cite{BertheDeFeliceDolceLeroyPerrinReutenauerRindone2013m}.

We finally prove a result which allows one to obtain a large family of neutral sets of geometric origin, namely using interval exchange transformations.
More precisely, we prove that the natural coding of an interval exchange transformation without connections of length $\geq 1$ is a neutral set.
This extends a result in~\cite{bifixcodesandintervalexchanges} concerning interval exchange without connections as well as a result of~\cite{BertheDelecroixDolcePerrinReutenauerRindone2014} concerning linear involutions without connetion.

\paragraph{Acknowledgement.} This work was supported by grants from R\'egion \^Ile-de-France and ANR project Eqinocs.

\section{Extension graphs}
\label{sec:graph}
Let $A$ be a finite alphabet.
We denote by $A^*$ the set of all words on $A$.
We denote by $\varepsilon$ or $1$ the empty word.
A set of words on the alphabet $A$ is said to be \emph{factorial} if it contains the factors of its elements. 
An \emph{internal factor} of a word $x$ is a word $v$ such that $x = uvw$ with $u,w$ nonempty.

Let $S$ be a factorial set on the alphabet $A$.
For $w\in S$, we denote $L_S(w) = \{ a \in A \mid aw \in S \}, \quad R_S(w) = \{ a \in A \mid wa \in S \}, E_S(w) = \{ (a,b) \in A \times A \mid awb \in S\}$, and further
$\ell_S(w) = \Card(L_S(w))$,  $r_S(w) = \Card(R_S(w))$,  $e_S(w) = \Card(E_S(w))$.

We omit the subscript $S$ when it is clear from the context.
A word $w$ is \emph{right-extendable} if $r(w)>0$, \emph{left-extendable} if $\ell(w)>0$ and \emph{biextendable} if $e(w)>0$.
A factorial set $S$ is called \emph{right-extendable} (resp. \emph{left-extendable}, resp. \emph{biextendable}) if every word in $S$ is right-extendable (resp. left-extendable, resp. biextendable).

A word $w$ is called \emph{right-special} if $r(w) \geq 2$.
It is called \emph{left-special} if $\ell(w) \geq 2$. 
It is called \emph{bispecial} if it is both left-special and right-special.
For $w\in S$, we denote 
$$
m_S(w) = e_S(w) - \ell_S(w) - r_S(w) + 1.
$$

A word $w$ is called \emph{neutral} if $m_S(w) = 0$.
We say that a set $S$ is \emph{neutral} if it is factorial and every nonempty word $w \in S$ is neutral. The \emph{characteristic} of $S$ is the integer
$\chi(S) = 1 - m_S(\varepsilon)$.

Thus, a neutral set of characteristic $1$ is such that all words (including the empty word) are neutral.
This is what is called a neutral set in~\cite{acyclicconnectedandtreesets}.

The following example of a neutral set is from~\cite{acyclicconnectedandtreesets}.

\begin{example}
\label{ex:cassaigne}
Let $A = \{ a,b,c,d \}$ and let $\sigma$ be the morphism from $A^*$ into itself defined by 
$\sigma: a \mapsto ab,\ b \mapsto cda,\ c \mapsto cd,\ d \mapsto abc$.
Let $S$ be the set of factors of the infinite word $x = \sigma^\omega(a)$.
One has $S\cap A^2 = \{ ab,ac,bc,ca,cd,da \}$ and thus $m(\varepsilon) = -1$.
It is shown in~\cite{acyclicconnectedandtreesets} that every nonempty word is neutral.
Thus $S$ is neutral of characteristic $2$.
\end{example}

A set of words $S \neq \{ \varepsilon \}$ is \emph{recurrent} if it is factorial and for any $u,w \in S$, there is a $v \in S$ such that $uvw \in S$.
An infinite factorial set is said to be \emph{uniformly recurrent} if for any word $u \in S$ there is an integer $n \geq 1$ such that $u$ is a factor of any word of $S$ of length $n$.
A uniformly recurrent set is recurrent.

The \emph{factor complexity} of a factorial set $S$ of words on an alphabet $A$ is the sequence $p_n = \Card(S \cap A^n)$.
Let $s_n = p_{n+1} - p_n$ and $b_n = s_{n+1} - s_n$ be respectively the first and second order differences sequences of the sequence $p_n$. 

The following result is~\cite[Proposition 3.5]{Cassaigne1997} (see also \cite[Theorem 4.5.4]{BertheRigo2010}).

\begin{proposition}
\label{pro:cant}
Let $S$ be a factorial set on the alphabet $A$.
One has $b_n = \sum_{w \in S \cap A^n} m(w)$ and $s_n = \sum_{w \in S \cap A^n} (r(w)-1)$ for all $n \geq 0$.
\end{proposition}

One deduces easily from Proposition~\ref{pro:cant} the following result which shows that a neutral set has linear complexity.

\begin{proposition}
\label{prop:complexity}
The factor complexity of a neutral set on $k$ letters is given by $p_0 = 1$ and $p_n = n(k - \chi(S)) + \chi(S)$ for every $n \geq 1$.
\end{proposition}

Let $S$ be a biextendable set of words.
For $w \in S$, we consider the set $E(w)$ as an undirected graph on the set of vertices which is the disjoint union of $L(w)$ and $R(w)$ with edges the pairs $(a,b) \in E(w)$.
This graph is called the \emph{extension graph} of $w$.
We sometimes denote $1 \otimes L(w)$ and $R(w) \otimes 1$ the copies of $L(w)$ and $R(w)$ used to define the set of vertices of $E(w)$.
We note that since $E(w)$ has $\ell(w)+r(w)$ vertices
and $e(w)$ edges, the number $1-m_S(w)$ is the Euler characteristic of the graph $E(w)$.

A biextendable set $S$ is called a \emph{tree set} of characteristic $c$ if for any nonempty $w \in S$, the graph $E(w)$ is a tree and if $E(\varepsilon)$ is a union of $c$ trees (the definition of tree set in~\cite{acyclicconnectedandtreesets} corresponds to a tree set of characteristic $1$).
Note that a tree set of characteristic $c$ is a neutral set of characteristic $c$.

\begin{example}
\label{ex:cassaigne2}
Let $S$ be the neutral set of Example~\ref{ex:cassaigne}.
The graph $E(\varepsilon)$ is represented in Figure~\ref{fig:julien}.
It is acyclic with two connected components.
It is shown in~\cite{acyclicconnectedandtreesets} that the extension graph of any nonempty word is a tree.
Thus $S$ is a tree set of characteristic $2$.
\end{example}

\begin{figure}[hbt]
\centering
\gasset{Nadjust=wh,AHnb=0}

\begin{picture}(40,5)
\node(a)(0,10){$a$}
\node(a')(40,0){$a$}
\node(b)(0,0){$b$}
\node(b')(10,10){$b$}
\node(c)(10,0){$c$}
\node(c')(30,10){$c$}
\node(d)(30,0){$d$}
\node(d')(40,10){$d$}

\drawedge(a,b'){}
\drawedge(a,c){}
\drawedge(b,c){}
\drawedge(c',a'){}
\drawedge(c',d'){}
\drawedge(d,a'){}
\end{picture}
\caption{The graph $E(\varepsilon$).}
\label{fig:julien}
\end{figure}
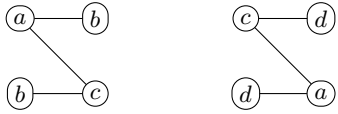

Let $S$ be a factorial set containing the alphabet $A$.
For $x \in S$, we define
$$
\rho_S(x) = e_S(x) - \ell_S(x), \quad \lambda_S(x) = e_S(x)-r_S(x).
$$
Thus, when $x$ is neutral, $\rho_S(x) = r_S(x)-1$ and $\lambda_S(x) = \ell_S(x)-1$.
The following result shows that in a biextendable neutral set, $\rho_S$ is a left probability distribution on $S$ (and $\lambda_S$ is a right probability), except for the value on $\varepsilon$ which 
is $\rho(\varepsilon)=e(\varepsilon)-\ell(\varepsilon) = m(\varepsilon)+r(\varepsilon)-1=\Card(A)-\chi(S)$
and can be different of $1$ (see~\cite{BerstelDeFelicePerrinReutenauerRindone2012} for the definition of a right or left probability distribution).
We omit the subscript $S$ when it is clear from the context.

\begin{proposition}
\label{pro:proba}
Let $S$ be a biextendable neutral set containing $A$.
Then for any $x \in S$, one has $\lambda_S(x), \rho_S(x) \geq 0$ and
$$
\sum_{a \in L(x)} \rho_S(ax) = \rho_S(x), \quad \sum_{a \in R(x)} \lambda_S(xa) = \lambda_S(x).
$$
\end{proposition}
\begin{proof}
Since $S$ is  biextendable, we have $\ell(x), r(x) \leq e(x)$.
Thus $\lambda(x), \rho(x) \geq 0$.
Next,
$\sum_{a \in L(x)}\rho(ax) = \sum_{a \in L(x)}(r(ax)-1) = e(x) - \ell(x) = \rho(x)$.
The proof for $\lambda$ is symmetric.
\end{proof}

If $\rho(\varepsilon) = 0$, then $\rho(x) = 0$ for all $x \in S$.
Otherwise, $\rho'(x) = \rho(x) / \rho(\varepsilon)$ is a left probability distribution.
A symmetric result holds for $\lambda$.

\section{Bifix codes}
\label{sec:bifix}
A prefix code is a set of nonempty words which does not contain any proper prefix of its elements.
A suffix code is defined symmetrically.
A \emph{bifix code} is a set which is both a prefix code and a suffix code (see~\cite{BerstelPerrinReutenauer2009} for a more detailed introduction).
Let $S$ be a recurrent set.
A prefix (resp. bifix) code $X \subset S$ is $S$-maximal if it is not properly contained in a prefix (resp. bifix) code $Y \subset S$.
Since $S$ is recurrent, a finite $S$-maximal bifix code is also an $S$-maximal prefix code (see~\cite[Theorem 4.2.2]{BerstelDeFelicePerrinReutenauerRindone2012}).
For example, for any $n \geq 1$, the set $X = S \cap A^n$ is an $S$-maximal bifix code.

Given a set $X$, we denote $\rho(X) = \sum_{x \in X}\rho(x)$.
We prove the following result.
It accounts for the fact that, in a Sturmian set $S$, any finite $S$-maximal suffix codes contains exactly one right-special word \cite[Proposition 5.1.5]{BerstelDeFelicePerrinReutenauerRindone2012}.

\begin{proposition}
\label{pro:maxsuf}
Let $S$ be a neutral set containing $A$, and let $X$ be a finite $S$-maximal suffix code.
Then $\rho(X) = \Card(A) - \chi(S)$.
\end{proposition}
\begin{proof}
If $\rho(\varepsilon) = 0$, then $\chi(S) = \Card(A)$ and thus the formula holds.
Otherwise, $\rho'$ is a left probability distribution (as seen at the end of Section~\ref{sec:graph}), and the formula holds by a well-known property of suffix codes (see~\cite[Proposition 3.3.4]{BerstelDeFelicePerrinReutenauerRindone2012}).
\end{proof}

\begin{example}
Let $S$ be the neutral set of characteristic $2$ of Example~\ref{ex:cassaigne}.
The set $X = \{ a,ac,b,bc,d \}$ is an $S$-maximal suffix code (its reversal is the $\tilde{S}$-maximal prefix code $\tilde{X} = \{ a,b,ca,cb,d \}$).
The values of $\rho$ on $X$ are represented in Figure~\ref{fig:maxsuf} on the left.
One has $\rho(X) = \rho(a) + \rho(bc) = 2$, in agreement with Proposition~\ref{pro:maxsuf}.
\end{example}

\begin{figure}[hbt]
\centering
\gasset{AHnb=0,Nadjust=wh}
\begin{picture}(110,30)(0,-2)

\put(0,0){
\begin{picture}(20,30)
\node(0)(20,15){$2$}
\node[Nmr=0](a)(10,30){$1$}
\node[Nmr=0](b)(10,20){$0$}
\node(c)(10,10){$1$}
\node[Nmr=0](d)(10,0){$0$}
\node[Nmr=0](ac)(0,15){$0$}
\node[Nmr=0](bc)(0,5){$1$}

\drawedge(a,0){$a$}
\drawedge(b,0){$b$ \ }
\drawedge(c,0){$c$}
\drawedge[ELside=r](d,0){$d$}
\drawedge(ac,c){$a$}
\drawedge[ELside=r](bc,c){$b$}
\end{picture}
}

\put(38,0){
\begin{picture}(20,30)
\node(0)(0,15){$2$}
\node(a)(10,28){$1$}
\node(b)(10,15){$0$}
\node[Nmr=0](c)(10,8){$1$}
\node(d)(10,0){$0$}
\node[Nmr=0](ab)(20,30){$0$}
\node(ac)(20,24){$0$}
\node(bc)(20,15){$1$}
\node[Nmr=0](da)(20,0){$1$}
\node[Nmr=0](acd)(30,24){$0$}
\node[Nmr=0](bca)(30,18){$0$}
\node[Nmr=0](bcd)(30,12){$0$}

\drawedge(0,a){$a$}
\drawedge(0,b){$b$}
\drawedge(0,c){$c$}
\drawedge[ELside=r](0,d){$d$}
\drawedge(a,ab){$b$}
\drawedge[ELside=r](a,ac){$c$}
\drawedge(b,bc){$c$}
\drawedge(d,da){$a$}
\drawedge(ac,acd){$d$}
\drawedge(bc,bca){$a$}
\drawedge[ELside=r](bc,bcd){$d$}
\end{picture}
}

\put(80,0){
\begin{picture}(20,30)
\node(0)(30,15){$2$}
\node(a)(20,27){$1$}
\node(b)(20,19){$0$}
\node[Nmr=0](c)(20,11){$1$}
\node(d)(20,3){$0$}
\node[Nmr=0](ca)(10,30){$0$}
\node(da)(10,24){$1$}
\node(ab)(10,19){$0$}
\node[Nmr=0](cd)(10,3){$0$}
\node[Nmr=0](bca)(0,30){$0$}
\node[Nmr=0](acd)(0,6){$0$}
\node[Nmr=0](bcd)(0,0){$0$}

\drawedge[ELside=r](0,a){$a$}
\drawedge[ELside=r](0,b){$b$ \ \ }
\drawedge[ELside=r](0,c){$c$ \ }
\drawedge(0,d){$d$}
\drawedge(a,da){$d$}
\drawedge[ELside=r](a,ca){$c$}
\drawedge(b,ab){$a$}
\drawedge[ELside=r](d,cd){$c$}
\drawedge[ELside=r](ca,bca){$b$}
\drawedge[ELside=r](cd,acd){$a$}
\drawedge(cd,bcd){$b$}
\end{picture}
}
\end{picture}
\caption{An $S$-maximal suffix code (left) and an $S$-maximal bifix code represented as a prefix code (center) and as a suffix code (right).}
\label{fig:maxsuf}
\end{figure}
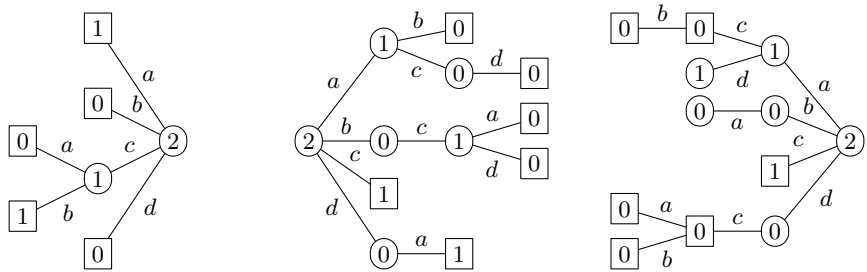

Let $X$ be a bifix code.
Let $Q$ be the set of words without any suffix in $X$ and let $P$ be the set of words without any prefix in $X$.
A \emph{parse} of a word $w$ with respect to a bifix code $X$ is a triple $(q,x,p) \in Q \times X^* \times P$ such that $w = qxp$.
We denote by $d_X(w)$ the number of parses of a word $w$ with respect to $X$.
The $S$-degree of $X$, denoted $d_X(S)$ is the maximal number of parses with respect to $X$ of a word of $S$.
For example, the set $X = S \cap A^n$ has $S$-degree $n$.

\begin{example}
\label{ex:maxbifix}
Let $S$ be the neutral set of characteristic $2$ of Example~\ref{ex:cassaigne}.
The set $X = \{ ab,acd,bca,bcd,c,da \}$ is an $S$-maximal bifix code of $S$-degree $2$ (see Figure~\ref{fig:maxsuf} on the center and the right).
\end{example}

Let $S$ be a recurrent set and let $X$ be a finite bifix code.
By~\cite[Theorem 4.2.8]{BerstelDeFelicePerrinReutenauerRindone2012}, $X$ is $S$-maximal if and only if its $S$-degree is finite.
Moreover, in this case, a word $w \in S$ is such that $d_X(w) < d_X(S)$ if and only if it is an internal factor of a word of $X$.
The following is~\cite[Theorem 4.3.7]{BerstelDeFelicePerrinReutenauerRindone2012}.

\begin{theorem}
\label{theo:union}
Let $S$ be a recurrent set and let $X$ be a finite $S$-maximal bifix code of $S$-degree $n$.
The set of nonempty proper prefixes of $X$ is a disjoint union of $n-1$ $S$-maximal suffix codes.
\end{theorem}

\begin{example}
Let $S$ and $X$ be as in Example \ref{ex:maxbifix}.
The set of nonempty proper prefixes of $X$ is the $S$-maximal suffix code represented on the left of Figure~\ref{fig:maxsuf}.
\end{example}

The following statement is closely related with a similar statement concerning the average length of a bifix code, but which requires an invariant probability distribution (see~\cite[Corollary 4.3.8]{BerstelDeFelicePerrinReutenauerRindone2012}).

\begin{proposition}
\label{pro:probap}
Let $S$ be a recurrent neutral set containing $A$, and let $X$ be a finite $S$-maximal bifix code of $S$-degree $n$.
The set $P$ of proper prefixes of $X$ satisfies $\rho_S(P) = n(\Card(A)-\chi(S))$.
\end{proposition}
\begin{proof}
By Theorem \ref{theo:union}, we have $P\setminus \{ \varepsilon \} = \cup_{i=1}^{n-1}Y_i$, where the $Y_i$ are $S$-maximal suffix codes.
By Proposition~\ref{pro:maxsuf}, we have $\rho(Y_i) = \Card(A) - \chi(S)$ and thus $\rho(P) = \rho(\varepsilon) + (n-1)(\Card(A) - \chi(S)) = n(\Card(A)-\chi(S))$.
\end{proof}

\section{Cardinality Theorem for bifix codes}
\label{sec:cardinality}
The following theorem is a generalization of~\cite[Theorem 3.6]{BertheDeFeliceDolceLeroyPerrinReutenauerRindone2013b} where it is proved for a neutral set of characteristic $1$.
We consider a recurrent set $S$ containing the alphabet $A$, and we implicitly assume that all words of $S$ are on the alphabet $A$.

\begin{theorem}
\label{theo:cardinality}
Let $S$ be a neutral recurrent set  containing the alphabet $A$.
For any finite $S$-maximal bifix code $X$ of $S$-degree $n$, one has
$$
\Card(X) = n(\Card(A)-\chi(S)) + \chi(S).
$$
\end{theorem}
Note that we recover, as a particular case of Theorem~\ref{theo:cardinality} applied to the set $X$ of words of length $n$ in $S$, the fact that for a set $S$ satisfying the hypotheses of the theorem, the factor complexity is $p_0 = 1$ and $p_n = n(\Card(A)-\chi(S))+\chi(S)$.
Note that Theorem~\ref{theo:cardinality} has a converse (see~\cite{BertheDelecroixDolceLeroyPerrinReutenauerRindone2015}).

\begin{proof}[of Theorem~\ref{theo:cardinality}]
Since $X$ is a finite $S$-maximal bifix code, it is an $S$-maximal prefix code (see Section~\ref{sec:bifix}).
By a well-known property of trees, this implies that $\Card(X) = 1 + \sum_{p\in P}(r(p)-1)$ where $P$ is the set of proper prefixes of $X$.
Since $\rho(p) = r(p)-1$ for $p$ non empty and $\rho(\varepsilon) = m(\varepsilon) + r(\varepsilon)-1$, we have
\begin{eqnarray*}
\Card(X)	& = &	1+\sum_{p\in P}(r(p)-1) = 1+\sum_{p\in P}\rho(p)-m(\varepsilon) \\
		& = &	\rho(P)+\chi(S)= n(\Card(A)-\chi(S))+\chi(S)
\end{eqnarray*}
since $\rho(P)=n(\Card(A)-\chi(S))$ by Proposition~\ref{pro:probap}.
\end{proof}

\begin{example}
Let $S$ be the neutral set of Example~\ref{ex:cassaigne} and let $X$ be the $S$-maximal bifix code of Example~\ref{ex:maxbifix}.
We have $\Card(X) = 2 (4-2) + 2 = 6$ according to Theorem~\ref{theo:cardinality}.
\end{example}

\section{Cardinality Theorem for return words}
\label{sec:return}

Let $S$ be a factorial set of words.
For a set $X \subset S$ of nonempty words, a \emph{complete first return word} to $X$ is a word of $S$ which has a proper prefix in $X$, a proper suffix in $X$ and no internal factor in $X$.
We denote by $\CR_S(X)$ the set of complete first return words to $X$.
The set $\CR_S(X)$ is a bifix code.
If $S$ is uniformly recurrent, $\CR_S(X)$ is finite for any finite set $X$.
For $x \in S$, we denote $\CR_S(x)$ instead of $\CR_S(\{ x \})$.

\begin{theorem}
\label{theo:completereturn}
Let $S$ be a  uniformly recurrent neutral set containing the alphabet $A$.
For any bifix code $X \subset S$, we have
$$
\Card(\CR_S(X)) = \Card(X)+\Card(A)-\chi(S).
$$
\end{theorem}
\begin{proof}
Let $P$ be the set of proper prefixes of $\CR_S(X)$.
For $q \in P$, we denote $\alpha(q) = \Card\{a \in A\mid qa \in P \cup \CR_S(X)\}-1$ and $\alpha(P) = \sum_{q\in P}\alpha(p)$.

Since $\CR_S(X)$ is a finite nonempty prefix code, we have, by a well-known property of trees, $\Card(\CR_S(X)) = 1+\alpha(P)$.

Let $P'$ be the set of words in $P$ which are proper prefixes of $X$
and let $Y=P\setminus P'$.
Since $P'$ is the set of proper prefixes of $X$, we have $\alpha(P) = \Card(X)-1$.

Since $S$ is recurrent, any word of $S$ with a prefix in $X$ is comparable for the prefix order with a word of $\CR_S(X)$.
This implies that for any $q \in Y$ and any $b \in R_S(q)$, one has $qb \in P \cup \CR_S(X)$.
Consequently, we have $\alpha(q) = \rho_S(q)$ for any $q \in Y$.
Thus we have shown that
$$
\Card(\CR_S(X)) = 1+\alpha(P')+\rho(Y) = \Card(X)+\rho(Y).
$$
Let us show that $Y$ is an $S$-maximal suffix code.
This will imply our conclusion by Proposition~\ref{pro:maxsuf}.
Suppose that $q, uq \in Y$ with $u$ nonempty.
Since $q$ is in $Y$, it has a proper prefix in $X$.
But this implies that $uq$ has an internal factor in $X$, a contradiction.
Thus $Y$ is a suffix code.
Consider $w \in S$.
Since $S$ is recurrent, there is some $u$ and $x \in X$ such that $xuw \in S$.
Let $y$ be the shortest suffix of $xuw$ which has a proper prefix in $X$.
Then $y \in Y$.
This shows that $Y$ is an $S$-maximal suffix code.
\end{proof}

Let $S$ be a factorial set.
A \emph{right first return word} to $x$ in $S$ is a word $w$ such that $xw$ is a word of $S$ which ends with $x$ and has no internal factor equal to $x$ (thus $xw$ is a complete first return word to $x$).
We denote by $\RR_S(x)$ the set of right first return words to $x$ in $S$. Since $\CR_S(x)=x\RR_S(x)$,
the sets $\CR_S(x)$ and $\RR_S(x)$ have the same number of elements. Thus we have the
following consequence of Theorem~\ref{theo:completereturn}.

\begin{corollary}
Let $S$ be a uniformly recurrent neutral set containing $A$.
For any $x \in S$, the set $\RR_S(x)$ has $\Card(A) - \chi(S)+1$ elements.
\end{corollary}

\begin{example}
Consider again the neutral set $S$ of Example~\ref{ex:cassaigne}.
We have $\RR_S(a) = \{ bca,bcda,cad \}$.
\end{example}

\section{Bifix decoding}
\label{sec:decoding}
Let $S$ be a factorial set and let $X$ be a finite $S$-maximal bifix code.
A \emph{coding morphism} for $X$ is a morphism $f : B^*\rightarrow A^*$ which maps bijectively an alphabet $B$ onto $X$.
The set $f^{-1}(S)$ is called a \emph{maximal bifix decoding} of $S$.

\begin{theorem}
\label{theo:decoding}
Any maximal bifix decoding of a recurrent neutral set is a neutral set
with the same characteristic.
\end{theorem}

Let $S$ be a factorial set.
For two sets of words $X,Y$ and a word $w \in S$, we denote $L_S^X(w) = \{x \in X \mid xw \in S \}, \ R_S^Y(w) = \{y \in Y \mid wy \in S\}, \ E_S^{X,Y}(w) = \{ (x,y) \in X \times Y \mid xwy \in S \}$,
and further
$$
e_S^{X,Y}(w) = \Card(E_S^{X,Y}(w)),\ \ell_S^X(w) = \Card(L_S^X(w)),\ r_S^Y(w) = \Card(R_S^Y(w)).
$$
Finally, for a word $w$, we denote $m_S^{X,Y}(w) = e_S^{X,Y}(w) - \ell_S^X(w) - r_S^Y(w)+1$.
Note that $E_S^{A,A}(w) = E_S(w)$, $m_S^{A,A}(w) = m_S(w)$, and so on.

\begin{proposition}
\label{pro:strong}
Let $S$ be a neutral set, let $X$ be a finite $S$-maximal suffix code and let $Y$ be a finite $S$-maximal prefix code.
Then $m_S^{X,Y}(w) =m_S(w)$ for every $w\in S$.
\end{proposition}
\begin{proof}
We may assume that $S$ contains the alphabet $A$.
We use an induction on the sum of the lengths of the words in $X$ and in $Y$. 

If $X,Y$ contain only words of length $1$, since $X$ (resp. $Y$) is an $S$-maximal suffix (resp. prefix) code, we have $X = Y = A$ and there is nothing to prove.

Assume next that one of them, say $Y$, contains words of length at least $2$.
Let $p$ be a nonempty proper prefix of $Y$ of maximal length.
Set $Y' = (Y\setminus pA) \cup p$.
If $wp \notin S$, then $m^{X,Y}(w) = m^{X,Y'}(w)$ and the conclusion follows by induction hypothesis.
Thus we may assume that $wp \in S$.
Then
\begin{displaymath}
m^{X,Y}(w)-m^{X,Y'}(w)= e^{X,A}(wp)- \ell^{X}(wp)-r^{A}(wp)+1=m^{X,A}(wp).
\end{displaymath}
By induction hypothesis, we have $m^{X,Y'}(w) = m(w)$
and $ m^{X,A}(wp)=0$, whence the conclusion.
\end{proof}

\begin{proof}[of Theorem~\ref{theo:decoding}]
Let $S$ be a recurrent neutral set and let $f : B^*\rightarrow A^*$ be a coding morphism for a finite $S$-maximal bifix code $X$.
Set $U=f^{-1}(S)$.
Let $v \in U \setminus \{ \varepsilon \}$ and let $w = f(v)$.
Then $m_U(v) = m_S^{X,X}(w)$.
Since $S$ is recurrent, $X$ is an $S$-maximal suffix code and prefix code.
Thus, by Proposition~\ref{pro:strong}, $m_U(v) = m_S(w)$,
which implies our conclusion.
\end{proof}

The following example shows that the maximal decoding of a uniformly recurrent neutral set need not be recurrent.

\begin{example}
Let $S$ be the set of factors of the infinite word $(ab)^\omega$.
The set $X = \{ ab, ba \}$ is a bifix code of $S$-degree $2$.
Let $f : u \mapsto ab$, $v \mapsto ba$.
The set $f^{-1}(S)$ is the set of factors of $u^\omega \cup v^\omega$ and it is not recurrent.
\end{example}

\section{Neutral sets and interval exchanges}
\label{sec:iet}
Let $I = ]\ell,r[$ be a nonempty open interval of the real line and $A$ a finite ordered alphabet.
For two intervals $\Delta, \Gamma$, we denote $\Delta < \Gamma$ if $x < y$ for any $x \in \Delta$ and $y \in \Gamma$.
A partition $(I_a)_{a \in A}$ of $I$ (minus $\Card(A)-1$ points) in open intervals is \emph{ordered} if $a < b$ implies $I_a < I_b$.

We consider now two total orders $<_1$ and $<_2$ on $A$ and two partitions $(I_a)_{a \in A}$ and $(J_a)_{a \in A}$ of $I$ in open intervals ordered respectively by $<_1$ and $<_2$ and such that for every $a$, $I_a$ and $J_a$ have the same length $\lambda_a$.
Let $\gamma_a = \sum_{b <_1 a} \lambda_b$ and $\delta_a = \sum_{b <_2 a} \lambda_a$.

An \emph{interval exchange transformation} (with flips) relative to $(I_a)_{a \in A}$ and $(J_a)_{a \in A}$ is a map $T : I \to I$ such that for every $a \in A$, its restriction to $I_a$ is either a translation or a symmetry from $I_a$ to $J_a$ (see, for example~\cite{bifixcodesandintervalexchanges} and \cite{NogueiraPiresTroubetzkoy2013} for interval exchanges with flips).

Observe that $\gamma_a$ is the left boundary of $I_a$ and that $\delta_a$ is the left boundary of $J_a$.
If $\Card(A) = s$, we say that $T$ is an $s$-interval exchange transformation.

\begin{example}
\label{ex:3iet}
Let $A = \{ a,b,c \}$.
Consider the rotation of angle $\alpha$ with $\alpha$ irrational as a $3$-transformation relative to the partition $(I_a)_{a \in A}$ of the interval $]0,1[$, where
$I_a = ]0, 1-2\alpha[$, $I_b = ]1-2\alpha, 1-\alpha[$ and $I_c = ]1-\alpha, 1[$, while $J_c = ]0, \alpha[$, $J_a = ]\alpha, 1-\alpha[$ and $J_b = ]1-\alpha, 1[$ (see Figure~\ref{fig:3iet}).
Then, for each letter $a$, the restriction to $I_a$ is a translation to $J_a$.
Note that one has $a <_1 b <_1 c$ and $c <_2 a <_2 b$.

\begin{figure}[hbt]
\centering
\gasset{Nh=2,Nw=2,ExtNL=y,NLdist=2,AHnb=0,ELside=r}
\begin{picture}(100,10)
\node(0h)(0,10){$0$}
\node(1-2alpha)(23.6,10){$1-2\alpha$}
\node(1-alpha)(61.8,10){$1-\alpha$}
\node(1h)(100,10){$1$}
\drawedge[linecolor=red,linewidth=1](0h,1-2alpha){$a$}
\drawedge[linecolor=blue,linewidth=1](1-2alpha,1-alpha){$b$}
\drawedge[linecolor=green,linewidth=1](1-alpha,1h){$c$}

\node(0b)(0,0){$0$}
\node(alpha)(38.2,0){$\alpha$}
\node(1-alphab)(61.8,0){$1-\alpha$}
\node(1b)(100,0){$1$}
\drawedge[linecolor=green,linewidth=1](0b,alpha){$c$}
\drawedge[linecolor=red,linewidth=1](alpha,1-alphab){$a$}
\drawedge[linecolor=blue,linewidth=1](1-alphab,1b){$b$}

\drawedge[AHnb=0,dash={0.2 0.5}0](1-alpha,1-alphab){}
\end{picture}
\caption{A $3$-interval exchange transformation.}
\label{fig:3iet}
\end{figure}
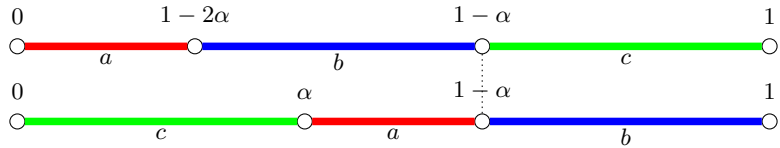
\end{example}
For a word $w = b_0 b_1 \cdots b_m$ let $I_w$ be the set
$$
I_w = I_{b_0} \cap T^{-1} \left( I_{b_1} \right) \cap \cdots \cap T^{-m} \left( I_{b_m} \right).
$$
Set $J_w = T^{|w|} \left( I_w \right)$.
We set by convention $I_\varepsilon = J_\varepsilon = ]\ell, r[$.
Note that each $I_w$ is an open interval and so is each $J_w$ (see~\cite{bifixcodesandintervalexchanges}).

Let $T$ be an interval exchange transformation on $I = ]\ell,r[$.
For a given $z \in I$, the \emph{natural coding} of $T$ relative to $z$ is the infinite word $\Sigma_T(z) = a_0 a_1 \cdots$ on the alphabet $A$ defined by $a_n = a$ if $T^n(z) \in I_a$.
We denote by $\mathcal{L}(T)$ the set of factors of the natural codings of $T$.
We also say that $\mathcal{L}(T)$ is the \emph{natural coding} of $T$.
Note that, for every $w \in \mathcal{L}(T)$, the interval $I_w$ is the set of points $z$ such that $\Sigma_T(z)$ starts with $w$, while the interval $J_w$ is the set of points $z$ such that $\Sigma_T\left( T^{-|w|}(z) \right)$ starts with $w$.
Moreover, it is easy to prove that a word $u$ is in $\mathcal{L}(T)$ if and only if $I_u \neq \emptyset$ (and thus if and only if $J_u \neq \emptyset$).

\begin{example}
Let $T$ be the interval exchange transformation of Example~\ref{ex:3iet}.
The first element of $\mathcal{L}(T)$ are represented in Figure~\ref{fig:factors} (right-special words are colored).

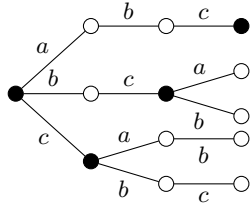
\begin{figure}[hbt]
\centering
\gasset{Nadjust=wh,AHnb=0}
\begin{picture}(30,21)
\node[fillcolor=black](e)(0,12){}
\node(a)(10,21){}
\node(b)(10,12){}
\node[fillcolor=black](c)(10,3){}
\node(ab)(20,21){}
\node[fillcolor=black](bc)(20,12){}
\node(ca)(20,6){}
\node(cb)(20,0){}
\node[fillcolor=black](abc)(30,21){}
\node(bca)(30,15){}
\node(bcb)(30,9){}
\node(cab)(30,6){}
\node(cbc)(30,0){}
\drawedge(e,a){$a$}
\drawedge(e,b){$b$}
\drawedge[ELside=r](e,c){$c$}
\drawedge(a,ab){$b$}
\drawedge(b,bc){$c$}
\drawedge(c,ca){$a$}
\drawedge[ELside=r](c,cb){$b$}
\drawedge(ab,abc){$c$}
\drawedge(bc,bca){$a$}
\drawedge[ELside=r](bc,bcb){$b$}
\drawedge[ELside=r](ca,cab){$b$}
\drawedge[ELside=r](cb,cbc){$c$}
\end{picture}
\caption{The words of length $\leq 3$ of $\mathcal{L}(T)$.}
\label{fig:factors}
\end{figure}
\end{example}

A \emph{connection} of an interval exchange transformation $T$ is a triple $(x,y,n)$ where $x$ is a singularity of $T^{-1}$, $y$ is a singularity of $T$, $n \geq 0$ and $T^n(x)=y$.
We also say that $(x,y,n)$ is a connection of length $n$ ending in $y$.
When $n=0$, we say that $x=y$ is a connection.

Interval exchange transformations without connections, also called \emph{regular} interval exchange transformations, are well studied (see, for example,~\cite{Keane1975} and~\cite{bifixcodesandintervalexchanges}).
The natural coding of a linear involutions without connection (see~\cite{BertheDelecroixDolcePerrinReutenauerRindone2014}) is essentially the coding of an interval exchange transformation with exactly one connection of length $0$ ending in the midpoint of the interval.

\begin{example}
Let $T$ be the transformation of Example~\ref{ex:3iet}.
The point $\gamma_c$ is a connection of length $0$.
This connection is represented with a dotted line in Figure~\ref{fig:3iet}.
\end{example}

Let $T$ be an interval exchange transformation with exactly $c$ connections all of length $0$.
Denote $\gamma_{k_0} = \ell$ and $\gamma_{k_1}, \ldots, \gamma_{k_c}$ the $c$ connections of $T$.
For every $0 \leq i < c$ the interval $]\gamma_{k_i}, \gamma_{k_{i+1}}[$ is called a \emph{component} of $I$.

\begin{example}
Consider again the transformation $T$ of Example~\ref{ex:3iet}.
The two components of $]0,1[$ are the two intervals $]0, 1-\alpha[$ and $]1-\alpha,1[$.
\end{example}

In the following result we generalize a result of~\cite{acyclicconnectedandtreesets} and show that the natural coding of an interval exchange is acyclic.

\begin{theorem}
\label{theo:quasitree}
Let $T$ be an interval exchange transformation with exactly $c$ connections, all of length $0$.
Then $\mathcal{L}(T)$ is neutral of characteristic $c$.
\end{theorem}

\begin{lemma}
\label{lem:awa}
Let $T$ be an interval exchange transformation.
For every nonempty word $w$ and letter $a \in A$, one has
\begin{enumerate}
\item[\rm (i)] $a \in L(w) \Longleftrightarrow I_w \cap J_a \neq \emptyset$,
\item[\rm (ii)] $a \in R(w) \Longleftrightarrow I_a \cap J_w \neq \emptyset$
\end{enumerate}
\end{lemma}
\begin{proof}
A letter $a$ is in the set $L(w)$ if and only if $a w \in \mathcal{L}(T)$.
As we have seen before, this is equivalent to $J_{aw} \neq \emptyset$.
One has $J_{aw} = T(I_{aw}) = T(I_a) \cap I_w = J_a \cap I_w$, whence point (i).
Point (ii) is proved symmetrically.
\end{proof}

We say that a path in a graph is \emph{reduced} if it does not use twice consecutively the same edge.

\begin{lemma}
\label{lem:components}
Let $T$ be an interval exchange transformation over $I$ without connection of length $\geq 1$.
Let $w \in \mathcal{L}(T)$ and $a,b \in L(w)$ (resp. $a,b \in R(w)$).
Then $1 \otimes a, 1 \otimes b$ (resp. $a \otimes 1, b \otimes 1$) are in the same connected component of $E(w)$ if and only if $J_a, J_b$ (resp. $I_a, I_b$) are in the same component of $I$.
\end{lemma}
\begin{proof}
Let $a \in L(w)$.
Since the set $\mathcal{L}(T)$ is biextendable, there exists a letter $c$ such that $(1 \otimes a, c \otimes 1) \in E(w)$.
Using the same reasoning that in Lemma~\ref{lem:awa}, one has $J_a \cap I_{wc} \neq \emptyset$.
Since $I_{wc} \subset I_w$, one has in particular $J_a \cap I_w \neq \emptyset$.
This proves that $J_a, I_w$ belong to the same component of $I$ for every $a \in L(w)$.

Conversely, suppose that $a,b \in L(w)$ are such that $J_a, J_b$ belong to the same component of $I$.
We may assume that $a <_2 b$.
Then, there is a reduced path $(1 \otimes a_1, b_1 \otimes 1, \ldots, b_{n-1} \otimes 1, 1 \otimes a_n)$ in $E(w)$ (see Figure~\ref{fig:path}) with $a=a_1$, $b=a_n$, $a_1 <_2 \cdots <_2 a_n$ and $wb_1 <_1 \cdots <_1 wb_{n_1}$.
Indeed, by hypothesis, we have no connection of length $\geq 1$.
Thus, for every $1 \leq i < n$, one has $J_{a_i} \cap I_{wb_i} \neq \emptyset$ and $J_{a_{i+1}} \cap I_{wb_i} \neq \emptyset$.
Therefore, $a,b$ are in the same connected component of $E(w)$.

The symmetrical statement is proved similarly.
\end{proof}

We can now proof the main result of this section.

\begin{proof}[of Theorem~\ref{theo:quasitree}]
Let us first prove that for any $w \in \mathcal{L}(T)$, the graph $E(w)$ is acyclic.
Assume that $(1 \otimes a_1, b_1 \otimes 1, \ldots, 1 \otimes a_n, b_n \otimes 1)$ is a reduced path in $E(w)$ with $a_1, \ldots, a_n \in L(w)$ and $b_1, \ldots, b_n \in R(w)$.
Suppose that $n \geq 2$ and that $a_1 <_2 a_2$.
Then one has $a_1 <_2 \cdots <_2 a_n$ and $wb_1 <_1 \cdots <_1 wb_n$ (see Figure~\ref{fig:path}).
Thus one cannot have an edge $(a_1, b_n)$ in the graph $E(w)$.

\begin{figure}[hbt]
\centering
\gasset{Nh=2,Nw=2,ExtNL=y,NLdist=2,AHnb=0}
\begin{picture}(100,12)(0,-2)

\node(1)(5,10){}
\node[Nframe=n](1a)(5,0){}
\node(2)(20,10){}
\node[Nframe=n](2a)(20,0){}
\node(3)(25,10){}
\node[Nframe=n](3a)(25,0){}
\node(4)(35,10){}

\node(5)(54,10){}
\node[Nframe=n](5a)(54,00){}
\node(6)(70,10){}
\node[Nframe=n](6a)(70,0){}
\node(7)(78,10){}
\node[Nframe=n](7a)(78,0){}
\node(8)(100,10){}

\put(40,5){$\cdots$}

\node(1b)(0,0){}
\node(2b)(13,0){}
\node[Nframe=n](2ba)(13,10){}
\node(3b)(17,0){}
\node[Nframe=n](3ba)(17,10){}
\node(4b)(27,0){}
\node[Nframe=n](4ba)(27,10){}

\node(5b)(47,0){}
\node(6b)(57,0){}
\node[Nframe=n](6ba)(57,10){}
\node(7b)(62,0){}
\node[Nframe=n](7ba)(62,10){}
\node(8b)(84,0){}
\node[Nframe=n](8ba)(84,10){}

\drawedge(1,2){$I_{wb_1}$}
\drawedge(3,4){$I_{wb_2}$}
\drawedge(5,6){$I_{wb_{n-1}}$}
\drawedge(7,8){$I_{wb_n}$}

\drawedge[ELside=r](1b,2b){$J_{a_1}$}
\drawedge[ELside=r](3b,4b){$J_{a_2}$}
\drawedge[ELside=r](5b,6b){$J_{a_{n-1}}$}
\drawedge[ELside=r](7b,8b){$J_{a_n}$}

\drawedge[AHnb=0,dash={0.2 0.5}0](1,1a){}
\drawedge[AHnb=0,dash={0.2 0.5}0](2,2a){}
\drawedge[AHnb=0,dash={0.2 0.5}0](2b,2ba){}
\drawedge[AHnb=0,dash={0.2 0.5}0](3,3a){}
\drawedge[AHnb=0,dash={0.2 0.5}0](3b,3ba){}
\drawedge[AHnb=0,dash={0.2 0.5}0](4b,4ba){}
\drawedge[AHnb=0,dash={0.2 0.5}0](5,5a){}
\drawedge[AHnb=0,dash={0.2 0.5}0](6,6a){}
\drawedge[AHnb=0,dash={0.2 0.5}0](6b,6ba){}
\drawedge[AHnb=0,dash={0.2 0.5}0](7,7a){}
\drawedge[AHnb=0,dash={0.2 0.5}0](7b,7ba){}
\drawedge[AHnb=0,dash={0.2 0.5}0](8b,8ba){}

\end{picture}
\caption{A path from $a_1$ to $a_n$ in $E(w)$.}
\label{fig:path}
\end{figure}
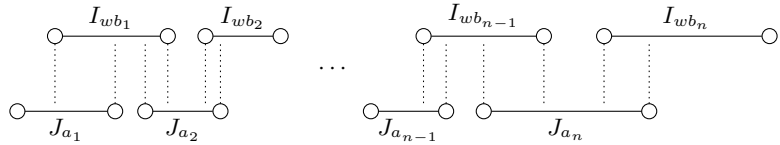

Let us now prove that the extension graph of the empty word is a union of $c$ trees.
Let $a,b \in A$.
If $J_a$ and $J_b$ are in the same component of $I$, then $1 \otimes a, 1\otimes b$ are in the same connected component of $E(\varepsilon)$ by Lemma~\ref{lem:components}.
Thus $E(\varepsilon)$ is a union of $c$ trees.

Finally, if $w \in \mathcal{L}(T)$ is a nonempty word and $a,b \in L(w)$, then $J_a$ and $J_b$ are in the same component of $I$, by Lemma~\ref{lem:awa}, and thus $a,b$ are in the same connected component of $E(w)$ by Lemma~\ref{lem:components}.
Thus $E(w)$ is a tree.
\end{proof}

The previous proof shows actually a stronger result: the set $\mathcal{L}(T)$ is a tree set of characteristic $c$.
This result generalizes the corresponding result for regular interval exchange in~\cite{acyclicconnectedandtreesets}.

\begin{example}
Let $T$ be the interval exchange transformation of Example~\ref{ex:3iet}.
In Figure~\ref{fig:trees} are represented the extension graphs of the empty word (left) and of the letters $a$ (center) and $b$ (right).
\end{example}

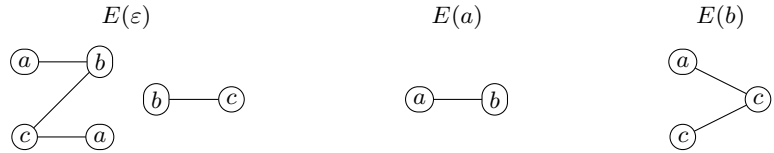
\begin{figure}[hbt]
\centering
\gasset{Nadjust=wh,AHnb=0,ELside=r}
\begin{picture}(97.5,17)
\node[Nframe=n](ee)(13.5,16){$E(\varepsilon)$}
\node(eal)(0,10){$a$}
\node(ecl)(0,0){$c$}
\node(ebr)(10,10){$b$}
\node(ear)(10,0){$a$}
\node(ebl)(17.5,5){$b$}
\node(ecr)(27.5,5){$c$}
\drawedge(eal,ebr){}
\drawedge(ecl,ebr){}
\drawedge(ecl,ear){}
\drawedge(ebl,ecr){}

\node[Nframe=n](ee)(57.5,16){$E(a)$}
\node(acl)(52.5,5){$a$}
\node(abr)(62.5,5){$b$}
\drawedge(acl,abr){}

\node[Nframe=n](ee)(92.5,16){$E(b)$}
\node(bal)(87.5,10){$a$}
\node(bcl)(87.5,0){$c$}
\node(bcr)(97.5,5){$c$}
\drawedge(bal,bcr){}
\drawedge(bcl,bcr){}
\end{picture}
\caption{Some extension graphs.}
\label{fig:trees}
\end{figure}

\newpage

\bibliographystyle{plain}
\bibliography{dlt15}

\end{document}